\documentclass[runningheads]{llncs}

\usepackage{amsmath}
\usepackage{latexsym}
\usepackage{amsfonts,amssymb}
\usepackage{graphicx, gastex}
\usepackage{cite}
\usepackage{mathrsfs}

\DeclareSymbolFont{rsfscript}{OMS}{rsfs}{m}{n}
\DeclareSymbolFontAlphabet{\mathrsfs}{rsfscript}

\DeclareMathOperator{\dt}{.}
\begin{document}

\title{Reset Complexity of Ideal Languages}

\author{Marina I. Maslennikova}

\titlerunning{Reset Complexity of Ideal Languages}
\authorrunning{M. Maslennikova}

\institute{Ural Federal University, Lenina st. 51, 620083, Ekaterinburg, Russia\\
\email{maslennikova.marina@gmail.com}}

\maketitle

\begin{abstract}
We present a new characteristic of a regular ideal language called reset complexity. We find some bounds on the reset complexity in terms of the state complexity of a given language. We also compare the reset complexity and the state complexity for languages related to slowly synchronizing automata and study uniqueness question for automata yielding the minimum of reset complexity.
\end{abstract}

\section{Introduction}

Let $\mathscr{A}=\langle Q,\Sigma,\delta\rangle$ be a \textit{deterministic finite automaton} (DFA), where $Q$ is the \textit{state set}, $\Sigma$ is the \textit{input alphabet}, and $\delta : Q \times \Sigma \rightarrow Q$ is the \textit{transition function} defining an action of the letters in $\Sigma$ on $Q$. The action extends in a unique way to an action $Q \times \Sigma^{*} \rightarrow Q$ of the free monoid $\Sigma^{*}$ over $\Sigma$; the latter action is still denoted by $\delta$. When function $\delta$ is clear from the context, we will write $q\dt w$ instead of $\delta(q,w)$ for $q\in Q$. For convenience we denote the set $\{\delta(q,w) \mid q \in S\}$ by $S \dt w$ for $S\subseteq Q$ and $w\in\Sigma^{*}$. In theory of formal languages the definition of a DFA usually includes the set $F\subseteq Q$ of \textit{terminal states} and an initial state \textit{initial state} $q_0\in Q$. We will use this definition when dealing with automata as devices for recognizing languages. The language $L\subseteq\Sigma^*$ is recognized by an automaton $\mathscr{A}=\langle Q,\Sigma,\delta, F, q_0\rangle$ if $L=\{w\in\Sigma^*\mid \delta(q_0,w)\in F\}$.

A DFA $\mathscr{A}=\langle Q,\Sigma,\delta\rangle$ is called \textit{synchronizing}, if there exists a word $w  \in \Sigma^{*}$ which leaves the automaton in one particular state no matter which state in $Q$ it starts at: $\delta(q,w)=\delta(q', w)$ for all $q, q' \in Q$. Any such word is said to be \textit{synchronizing} (\textit{reset}) for the DFA $\mathscr{A}$.
Synchronizing automata are of interest, motivated mostly by the \v{C}ern\'{y} conjecture. \v{C}ern\'{y} in \cite{Cerny} produced for each $n>1$ a synchronizing automaton $\mathscr{C}_{n}$ with $n$ states over a binary alphabet whose shortest synchronizing word has length $(n-1)^{2}$. Later he conjectured that any synchronizing automaton with $n$ states possesses a synchronizing word of length at most $(n-1)^{2}$. This conjecture has been proved for various classes of synchronizing automata, nevertheless in general it remains one of the most longstanding open problems in automata theory. \v{C}ern\'{y} series provides lower bound on maximum possible length of shortest synchronizing words for synchronizing automata with $n$ states. On the other hand, the best upper bound known so far is cubic: $\frac{n^{3}-n}{6}$ \cite{Pin}. For more details on synchronizing automata see the survey \cite{Vo_Survey}.

By $Syn(\mathscr{A})$ we denote the language of all words synchronizing $\mathscr{A}$. For a given DFA $\mathscr{A}=\langle Q,\Sigma,\delta, q_{0}, F\rangle$ the state $s$ is called \textit{reachable} if there exists a word $w\in \Sigma^{*}$ with the property $\delta(q_0, w)=s$, and \textit{unreachable}, otherwise. By $L(\mathscr{A}^{s})$ we denote the set of all words recognized by a DFA $\mathscr{A}^{s}=\langle Q,\Sigma,\delta, s, F\rangle$ which is obtained from $\mathscr{A}$ by choosing $s$ as the initial state. The states $s$ and $t$ of a DFA $\mathscr{A}$ are \textit{equivalent} if $L(\mathscr{A}^{s})=L(\mathscr{A}^{t})$. A DFA with a distinguished initial state and a distinguished set of final states is called \textit{minimal} if it contains no (different) equivalent states, and all states are reachable from the initial state. For a given regular language $L$ minimal automaton recognizing $L$ is unique up to isomorphism. The number of states in the minimal DFA is denoted by $sc(L)$ and is called \textit{state complexity} of the language $L$.

In what follows we consider only ideal languages, that is, languages $L$ satisfying the property $L=\Sigma^{*}L\Sigma^{*}$. It is obvious that the language of synchronizing words of a given synchronizing automaton satisfies this property. Now we prove the following
\begin{lemma}
\label{L Syn}
 Let $L$ be an ideal language and $\mathscr{A}$ the minimal automaton recognizing $L$. Then $\mathscr{A}$ is synchronizing and $Syn(\mathscr{A})=L$.
\end{lemma}
\begin{proof}

Note that, for any word $w \in L$ the word $wu \in L$ for all $u \in \Sigma^{*}$. Thus the minimal automaton $\mathscr{A}=\langle Q,\Sigma,\delta,q_0,F\rangle$ recognizing $L$ has only one terminal state $f$. Otherwise we would remove all transitions from terminal states and replace all terminal states by the unique terminal state $f$. Next we would put $\delta(f,a)=f$ for all $a\in\Sigma$. In such a way we would obtain a DFA recognizing the same language with less states than $\mathscr{A}$.

Now we prove that $L\subseteq Syn(\mathscr{A})$. Take any $w \in L$. By the definition $\delta(q_0, w) =f$. Let $q$ be an arbitrary state of the DFA $\mathscr{A}$. All states in $\mathscr{A}$ are reachable, thus there exists a word $u \in \Sigma^{*}$ such that $\delta (q_{0},u)=q$. Consider now the word $uw$. Note that $uw \in L$, hence $\delta (q_{0},uw)=f$. On the other hand, we have $\delta (q_{0},u)=q$, then $\delta (q,w)=f$. By the definition of a synchronizing word we get that, for any $w \in L$ holds $w \in Syn(\mathscr{A})$. Hence the DFA $\mathscr{A}$ is synchronizing, and $L\subseteq Syn(\mathscr{A})$.
Take now $w \in Syn(\mathscr{A})$. Clearly $w$ brings any state of $\mathscr{A}$ to $f$. In particular $\delta(q_0,w)=f$, i.e. $w \in L$. Hence $Syn(\mathscr{A}) \subseteq L$. And the equality $Syn(\mathscr{A}) = L$ holds.
\qed
\end{proof}

Lemma~\ref{L Syn} shows that for every ideal language $L$ there is a synchronizing automaton $\mathscr{A}$ such that $Syn(\mathscr{A})=L$.
 Thus, it is rather natural to find out how many states an automaton $\mathscr{A}$ may have. We define the \emph{reset complexity} $rc(L)$ of an ideal language $L$ as the minimal possible number of states in a synchronizing automaton $\mathscr{A}$ such that $Syn(\mathscr{A})=L$. By Lemma~\ref{L Syn} we have $rc(L)\leq sc(L)$. Now it is of interest how big a gap between $rc(L)$ and $sc(L)$ can be. Another interesting question concerns the uniqueness of the minimal in terms of reset complexity automaton. It is well-known that the minimal automaton recognizing a given language $L$ is unique up to isomorphism. One may think that the same fact holds for the synchronizing automaton minimal in terms of reset complexity. However, as our results show, in general this is not the case.

The notion of reset complexity might give a new approach to the \v{C}ern\'{y} conjecture. Let $\ell$ be the length of shortest words in $L$, and let $rc(L)=n$. If we had proved the inequality $n\geq \frac{\sqrt{\ell}}{c}$ (where $c$ is some constant value), we would obtain quadratic upper bound on $\ell$, namely $\ell\leq c^{2}n^{2}$. By lemma~\ref{L Syn} $L$ is the language of synchronizing words for some automaton. Then inequality $\ell\leq c^{2}n^{2}$ presents a quadratic upper bound on the length of shortest synchronizing word for a given synchronizing automaton, that is a major step towards the proof of the \v{C}ern\'{y} conjecture.

Minimal in terms of reset complexity automata are useful for compact representation of a given language. Indeed, let $L$ be an ideal language. It is accepted by the minimal DFA $\mathscr{P}$. Note that, simple operations such as checking whether a given word $w$ is in $L$ take polynomial in the length of $w$ time, namely $O(|w|)$. On the other hand, the automaton $\mathscr{P}$ has $sc(L)$ states and this number may be rather large. Let $\mathscr{A}$ be a synchronizing DFA such that $Syn(\mathscr{A})=L$ and such that $\mathscr{A}$ has $rc(L)$ states. Now checking the property $w\in L$ takes $O(|w|\cdot rc(L))$ time. It is slightly worse than $O(|w|)$. However, as our results show, sometimes $sc(L)$ is an exponential function of $rc(L)$. So in this case we obtain exponential economy in space that is needed to keep the corresponding DFA.

\section{\normalsize{Upper and lower bounds on reset complexity}}

In this section we show that the upper bound from Lemma~\ref{L Syn} on reset complexity of a given ideal language $L$ is tight. Also we find a simple lower bound on $rc(L)$ in terms of the length of the shortest word in $L$.
To this aim we introduce some auxiliary notions. Given a subset $S$ of $Q$, by ${\cal C}(S)$ we denote the set of all words \textit{stabilizing} $S$:
$${\cal C}(S)=\{w\in\Sigma^{*}|\quad S \dt w=S\}.$$ We make use of the following results from \cite{Prib}.

\begin{lemma}\label{stab}{\mdseries\cite[Lemma~1.]{Prib}}
Given a word $w\in \Sigma^{*}$ there exists an integer $\beta \geq 0$ such that the set $m(w)=Q \dt w^{\beta}$ is fixed by $w$. Moreover $m(w)$ is the largest subset of $Q$ with this property.
\end{lemma}

Let $k(w)$ be the least integer with the property $Q \dt w^{k(w)}=m(w)$.

\begin{lemma}\label{card}{\mdseries\cite[Lemma~2.]{Prib}}
Given a word $w \in \Sigma^{*}$
$$ k(w)\leq |Q|-|m(w)|.$$
\end{lemma}

\begin{proposition}
\label{sc rc equal}
Let $L$ be an ideal language over a unary alphabet $\Sigma$. Then $sc(L)=rc(L)=\ell+1$, where $\ell$ is the minimum length of words in $L$.
\end{proposition}
\begin{proof}
 Let $\Sigma=\{a\}$ and $a^{\ell}$  be the shortest word in $L$. By the definition $L=a^{\ell}\Sigma^{*}$. The language $L$ is accepted by the DFA $\mathscr{A}$ from Fig.~\ref{fig A}.
\begin{figure}[ht]
\begin{center}
  \begin{picture}(70,8)(0,5)
   \gasset{Nw=6,Nh=6,Nmr=3}
   \thinlines
   \node(A1)(10,10){$0$}
   \node(A2)(25,10){$1$}
   \node(A3)(55,10){$\ell$}
   \node[Nframe=n](A4)(40,10){$...$}
   \node[Nframe=n](Ai)(0,10){$$}
   \node[Nframe=n](At)(65,10){$$}
   \drawloop[loopangle=90](A3){$a$}
   \drawedge(A1,A2){$a$}
   \drawedge(A2,A4){$a$}
   \drawedge(A4,A3){$a$}
   \drawedge(Ai,A1){$$}
   \drawedge(A3,At){$$}
   \end{picture}
\end{center}
\caption{Automaton $\mathscr{A}$}
\label{fig A}
\end{figure}
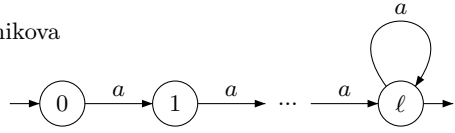
It is easy to see that the automaton $\mathscr{A}$ is minimal, so $sc(L)=\ell+1$. Now we verify that $rc(L)=\ell+1$. Let $\mathscr{B}=\langle Q,\Sigma,\delta\rangle$ be a DFA minimal in terms of reset complexity. Since the word $a^{\ell}$ is in $L=Syn(\mathscr{B})$, we have $H \dt a \neq H$ for any non-singleton subset $H \subseteq Q$ (otherwise $a^{\ell}$ would not be a synchronizing word). Hence we have $|m(a)|=1$, and $a^{k(a)}$ is a synchronizing word. It implies $k(a)=\ell$ and by lemma~\ref{card} we have $|Q|\ge \ell+1$. On the other hand, by lemma~\ref{L Syn} we have $rc(L)\leq sc(L)=\ell+1$, Hence the equality $|Q|=\ell+1$ holds.
\qed
\end{proof}

Proposition~\ref{sc rc equal} shows that the reset complexity and the state complexity of a given ideal language over a unary alphabet are equal. However, as we will see later in case of a binary alphabet an analogous statement is not true. Nevertheless the upper bound given by lemma~\ref{L Syn} is tight also in case $|\Sigma|=2$.
Consider the language $L_n=\Sigma^{*}a^{n-2}b\Sigma^{*}$. It is recognized by the DFA $\mathscr{A}_{n}$ from Fig.~\ref{fig An}. It is easy to see that $\mathscr{A}_{n}$ is the minimal DFA recognizing $L$. Thus $sc(L)=n$. Now we verify that $rc(L)=n$. Let $\mathscr{B}=\langle Q,\delta,\Sigma\rangle$ be a DFA minimal in terms of reset complexity. The word $a^{n-2}b$ is in $L=Syn(\mathscr{B})$. We have $Q \dt a^{i+1} \neq Q\dt a^{i}$ for any $0\leq i\leq n-3$ (otherwise $a^{i}b$ with $i<n-2$ would be synchronizing for $\mathscr{B}$, but $a^{i}b\not\in L$). Thus $k(a)\ge n-2$. Moreover $|m(a)|\ge 2$. Indeed, if $|m(a)|=1$, then the word $a^{k(a)}$ is synchronizing for $\mathscr{B}$, but $a^{k(a)}\notin L$. Thus by lemma~\ref{card} we have $|Q|\geq k(w)+|m(w)|\ge n-2+2=n$. On the other hand, by lemma~\ref{L Syn} $|Q|\leq sc(L)=n$. Hence the equality $|Q|=n$ holds.
{\par\bigskip}
\begin{figure}[ht]
\begin{center}
  \begin{picture}(105,10)(0,3)
   \gasset{Nw=6,Nh=6,Nmr=3}
   \thinlines
   \node[Nframe=n](Ai)(0,17){$$}
   \node[Nframe=n](At)(105,17){$$}
   \node(A0)(10,10){$0$}
   \node(A1)(25,10){$1$}
   \node(A2)(40,10){$2$}
   \node[Nadjust=wh,Nmr=1](A4)(70,10){$n-2$}
   \node[Nadjust=wh,Nmr=1](A5)(95,10){$n-1$}
   \node[Nframe=n](A3)(55,10){$...$}
   \drawloop[loopdiam=8,loopangle=-45](A5){$a,b$}
   \drawloop[loopdiam=8,loopangle=-90](A4){$a$}
   \drawedge(A0,A1){$a$}
   \drawedge(Ai,A0){$$}
   \drawedge(A5,At){$$}
   \drawloop[loopdiam=8,loopangle=-135](A0){$b$}
   \drawedge(A1,A2){$a$}
   \drawedge(A2,A3){$a$}
   \drawedge(A3,A4){$a$}
   \drawedge(A4,A5){$b$}
   \drawedge[ELside=r,curvedepth=-5](A1,A0){$b$}
   \drawedge[curvedepth=5](A2,A0){$b$}
   \end{picture}
\end{center}
\caption{Automaton $\mathscr{A}_{n}$}
\label{fig An}
\end{figure}
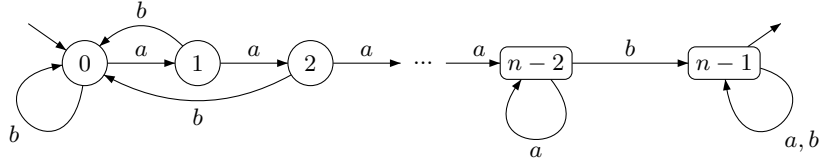

Examples above show that the upper bound $rc(L)\leq sc(L)$ is tight. A simple lower bound on $rc(L)$ can be given in terms of the length $\ell$ of the shortest word in $L$. It is based on the known upper bound on the length of the shortest synchronizing word for a given DFA. We have $rc(L)\geq f(n)=\min\{n\in \mathbb{N} | \frac{n^3-n}{6}\geq l\}$.

\section{\normalsize{Reset and state complexity of slowly synchronizing automata}}

Given a synchronizing automaton $\mathscr{A}=\langle Q,\Sigma,\delta\rangle$, the language $Syn(\mathscr{A})$ can be recognized by the power automaton ${\cal P}=\langle {\cal Q}, \Sigma,\delta,Q,F\rangle$, where ${\cal Q}$ is the set of all nonempty subsets of $Q$, the transition function is a natural extension of $\delta$ (still denoted by $\delta$), the set $Q$ plays the role of the initial state, and $F=\{H \subseteq Q\mid |H|=1\}$. In the examples below we show that for a given synchronizing automaton with $n$ states, its minimized power automaton has $2^{n}-n$ states. Using this result we prove that for a binary alphabet the gap between $rc(L)$ and $sc(L)$ can be exponentially large. Considered automata are examples of ``slowly'' synchronizing automata, i.e. automata whose shortest synchronizing words have length close to $(n-1)^{2}$. The first example belongs to \v{C}ern\'{y} \cite{Cerny}, the others are taken from \cite{AnGuVol}.

Let $\mathscr{A}=\langle Q,\Sigma,\delta\rangle$ be a DFA with $n$ states. Denote $\delta ^{-1}(H,w)=\{q\in Q \mid \delta(q, w)\in H\}$. We define the function $d(p,q):Q\times Q \rightarrow \mathbb{R}$ of \emph{distance} between states $p$ and $q$ as follows (without loss of generality assume that $p<q$): \begin{equation} \label{dist}d(p,q)=min\{q-p,n+p-q\}.\end{equation} The function $d(H)$ for a subset $H \subseteq Q$ is defined in a natural way, namely, \begin{equation} \label{diam}
d(H)=\min_{p,q\in H, p\neq q} d(p,q).
\end{equation}

Consider the \v{C}ern\'{y} automaton with $n$ states $\mathscr{C}_{n}$ (see Fig.~\ref{fig Cn}). Its transition function is defined as follows: \\
$i\dt b=i+1$ for $0\le i\le n-2$, and $(n-1)\dt b=0$;\\
$i\dt a=i$ for $0\le i\le n-2$, and $(n-1)\dt a=0$.
\begin{figure}[ht]
\begin{center}
  \unitlength=2pt
  \begin{picture}(50,75)(0,5)
   \gasset{Nw=10,Nh=10,Nmr=5}
   \thinlines
\node(A1)(0,60){$0$}
\node(A2)(30,75){$1$}
\node[Nframe=n](A3)(60,60){$\ldots$}
\node(n3)(60,25){$n$-$3$}
\node(n2)(30,10){$n$-$2$}
\node(n1)(0,25){$n$-$1$}
   \drawloop[loopdiam=8,loopangle=180](A1){$a$}
   \drawloop[loopdiam=8,loopangle=90](A2){$a$}
   \drawloop[loopdiam=8,loopangle=0](n3){$a$}
   \drawloop[loopdiam=8,loopangle=-90](n2){$a$}
   \drawedge(A1,A2){$b$}
   \drawedge(A2,A3){$b$}
   \drawedge(n2,n1){$b$}
   \drawedge(n1,A1){$a,b$}
   %\node[Nframe=n](p0)(50,23){$\ldots$}
   \drawedge(A3,n3){$b$}
   \drawedge(n3,n2){$b$}
   \end{picture}
   \end{center}
   \caption{\v{C}ern\'{y} automaton $\mathscr{C}_{n}$}
   \label{fig Cn}
\end{figure}
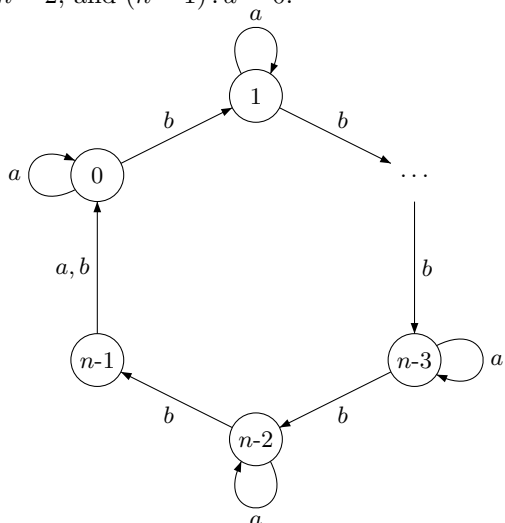

\begin{proposition}
\label{Cn sc}
$sc(Syn(\mathscr{C}_{n}))=2^{n}-n$.
\end{proposition}

\begin{proof}
From the \v{C}ern\'{y} automaton $\mathscr{C}_{n}=\langle Q,\Sigma,\delta\rangle$ construct its power automaton ${\cal P}=\langle {\cal Q}, \Sigma,\delta,Q,F\rangle$.

First we check that all nonempty subsets $H\subseteq Q$ are reachable. By induction on $k=|H|$. Case $|H|=n$ is clear: the state set $Q$ of  the automaton $\mathscr{C}_{n}$ is the initial state of ${\cal P}$. Assume that any subset with cardinality $1<k\leq n$ is reachable. Now we verify that all subsets $H$ with $|H|=k-1$ are reachable. Let $H=\{p_{1},p_{2},...,p_{k-1}\}$, and $p_{i}<p_{i+1}$ for all $1\leq i\leq k-2$.

If $H=\{0,1,...,k-2\}$, then $H$ is reachable from $H'=\{0,1,...,k-2,n-1\}$: clearly $H'\dt a=H$. Otherwise there exists a positive integer $i$ such that $p_{i} \neq i-1$. Let $\overline{i}$ be the least such an integer, so $H=\{0,1,\ldots,\overline{i}-2,p_{\overline{i}},\ldots,p_{k-1}\}$. Then $H$ is reachable from $$H'=\{0,p_{\overline{i}+1}-p_{\overline{i}},p_{\overline{i}+2}-p_{\overline{i}},
...,p_{k-1}-p_{\overline{i}},n-p_{\overline{i}},n+1-p_{\overline{i}},...,n+\overline{i}-2-p_{\overline{i}},n-1\}.$$ Indeed, it is not hard to see that after applying $ab^{p_{\overline{i}}}$ to $H'$ we obtain $H$.
Since $|H'|=k$, by induction hypothesis $H'$ is reachable. Thus there exists a word $w \in \Sigma ^{*}$ such that $Q \dt w=H'$, and we have $Q \dt wab^{p_{\overline{i}}}=H $.

So, the automaton ${\cal P}$ consists of $2^{n}-1$ reachable states. Obviously, all singletons are equivalent, thus the minimal automaton recognizing $Syn(\mathscr{C}_{n})$ has at most $2^{n}-n$ states.

Next we prove that any two states of ${\cal P}$ different from the terminal one are inequivalent. Take two arbitrary subsets $H$ and $S$ of $Q$ such that $H\ne S$. We verify that there exists a word $w \in \Sigma^{*}$ such that $H \dt w\neq S\dt w$ and at least one of the equalities $|H\dt w|=|H|-1$ and $|S\dt w|=|S|-1$ holds.

First consider a two-state subset $T=\{p,q\}$, $p<q$. Let $d=d(p,q)$, define the parameter $\alpha$ by the formula: $\alpha=n-p-1$ in case $d(p,q)=q-p$, and $\alpha=n-q-1$ in case $d(p,q)=n+p-q$. Note that, the word \begin{equation} \label{syn}w=b^{\alpha}a(b^{n-1}a)^{d-1}\end{equation} synchronizes $T$.
Indeed, $T\dt b^{\alpha}=\{n-1, d-1\}$ (depending on $\alpha$ we have either $p\dt b^{\alpha}=n-1$, and $q\dt b^{\alpha}=d-1$ or viceversa). Hence $T\dt b^{\alpha}a=\{0,d-1\}$. The state $0$ is fixed by the word $u=b^{n-1}a$, and the state $d-1$ under the action of the word $u$ moves to the state $d-2$. So, after applying the word $u$ to the set $\{0,d-1\}$ $d-1$ times, we obtain the set $\{0\}$. So, $T\dt w=\{0\}$.

Given two subsets $H$ and $S$ of $Q$ let us find $d=min\{d(H), d(S)\}$. We denote by $\{p,q\}$ the pair with the distance $d$ (if there are several such pairs, we construct for them corresponding words by the formula~\eqref{syn} and choose the pair with the shortest word $w$). Without loss of generality we may assume $\{p,q\}\subseteq H$. Next we prove the following auxiliary

\begin{claim}
Let $\{p,q\}$ be the pair chosen as above, and $w$ be the word constructed for the pair $\{p,q\}$ by the formula~\eqref{syn}. No other pair in $H$ or $S$ is synchronized by the word $w$.
\end{claim}
\begin{proof}
 Arguing by contradiction, suppose there is a pair $T'=\{p',q'\}$ ($p'<q'$) either in $H$ or in $S$ such that $p'\dt w=q'\dt w$ and either $p'$ is different from $p$ or $q'$ is different from $q$. Let $d'=d(p',q')$. By the definition of $d$ we have $d'\ge d$. Suppose $0\notin T'\dt b^{\alpha}a=\{p'',q''\}$. Then $p''\dt b^{n-1}a=p''-1$, and $q''\dt b^{n-1}a= q''-1$. Thus the distance between the states $p''$ and $q''$ does not change, so, for the word $w$ to synchronize $T'$ it is necessary that $p''=q''$. Since $b$ is a permutation letter, the only possibility for this to happen is $p'\dt b^\alpha a=q'\dt b^\alpha a=0$, a contradiction with the supposition $0\notin \{p'',q''\}$. So, $0\in T'\dt b^\alpha a$. If $n-1\in T'\dt b^\alpha$, then $T'\dt b^\alpha=\{n-1, d'-1\}$. If $d'>d$, then $T'\dt w=\{0,d'-d\}$, a contradiction. Thus $d'=d$, but then $T'\dt b^\alpha=T\dt b^\alpha$. Since $b$ is a permutation letter, we get $T=T'$. Since $p<q$ and $p'<q'$ we get $p=p'$, $q=q'$. A contradiction. So we have $0\in T'\dt b^\alpha$. Hence $T'\dt b^\alpha=\{0,d'\}$. But then $T'\dt w=\{0,d'-d+1\}$, and $d'-d+1>0$ even in case $d'=d$. Again a contradiction.
 \qed
 \end{proof}

Returning to the proof of the proposition consider three cases.

\textbf{Case 1:} $p,q\in H \setminus S$. If $S\subsetneq H$, then $|S|\le |H|-2$. By the claim we have $|H\dt w|=|H|-1$, and $|S\dt w|=|S|$. Thus $|S\dt w|=|S|\le |H|-2=|H\dt w|-1<|H\dt w|$. Therefore $S\dt w\ne H\dt w$. Suppose now $S\setminus H\ne \varnothing$. We show that either $(H\setminus S)$ and $(S\setminus H)$ do not intersect or $(H\setminus S)\dt w \cap (S\setminus H)\dt w =\{0\}$.
Since $b$ is a permutation letter, $(H\setminus S) \dt b^{\alpha}$ and $(S\setminus H) \dt b^{\alpha}$ do not intersect. By the definition of $\alpha$ we have $\{p,q\}\dt b^{\alpha}=\{n-1,d-1\}$, and $\{p,q\}\dt b^\alpha a=\{0,d-1\}$. Since $n-1\in (H\setminus S) \dt b^{\alpha}$ we have $n-1 \not\in (S\setminus H)  \dt b^{\alpha}$. Hence, $(S\setminus H) \dt b^{\alpha}a=(S\setminus H) \dt b^{\alpha}$. Thus the subsets $(H\setminus S)  \dt b^{\alpha}a$ and $(S\setminus H)  \dt b^{\alpha}a$ can have at most one common element, namely $0$.
Note that, for each $ r \in Q$ such that $r\neq 0$ we have $r\dt b^{n-1}a=r-1$, and for $r=0$ we have $r\dt b^{n-1}a=0$. Thus the word $b^{n-1}a$ shifts by $1$ all the states different from $0$. Moreover, through $d-1$ steps no state different from $p$ and $q$ moves to $0$, otherwise we would get another pair synchronized by $w$, which contradicts the claim. This implies that the subsets $(H\setminus S)\dt w$ and $(S\setminus H)\dt w$ can have at most one common element $0$. If $(H\setminus S)\dt w\ne \{0\}$ or $(S\setminus H)\dt w\ne \{0\}$, then obviously $H\dt w\ne S\dt w$. It remains to study the case when $(H\setminus S)\dt w=(S\setminus H)\dt w= \{0\}$. By the claim it is possible only if $H\setminus S=\{p,q\}$ and $S\setminus H=\{r\}$ for some $r\in Q$. Then we have $r\dt w=0$. By the definition of $w$ we get $r\dt b^{\alpha}a(b^{n-1}a)^{k}=0$ for some $0\leq k\leq d-2$. Consider subsets $\overline{H}=H \dt b^{\alpha}a(b^{n-1}a)^{d-2}$ and $\overline{S}=S \dt b^{\alpha}a(b^{n-1}a)^{d-2}$. Note that $\overline{S}\subseteq \overline{H}$ and $\overline{H}\setminus \overline{S}=\{1\}$. It remains to check that there exists a word $v \in \Sigma^{*}$ such that $\overline{H} \dt v\neq \overline{S}\dt v$ and at least one of the equalities $|\overline{H}\dt v|=|\overline{H}|-1$ and $|\overline{S}\dt v|=|\overline{S}|-1$ holds. This situation will be studied later inside the Case 3.

\textbf{Case 2:} $p,q\in I=H \cap S$. If $S\subsetneq H$, then $|S|<|H|$. By the claim $|S\dt w|=|S|-1$ and $|H\dt w|=|H|-1$, so $S\dt w\ne H\dt w$. The case $H\subsetneq S$ is considered symmetrically. So we may assume $S\setminus H\ne\varnothing$ and $H\setminus S\ne\varnothing$.
 We show that $(H\setminus S) \dt w \cap (S\setminus H) \dt w =\varnothing$. Apply $w$ to $H\setminus S$ and $S\setminus H$.
Since $b$ is a permutation letter, it is clear that $(H\setminus S) \dt b^{\alpha}$ and $(S\setminus H) \dt b^{\alpha}$ have empty intersection. Next we apply the letter $a$ to $(H\setminus S) \dt b^{\alpha}$ and $(S\setminus H) \dt b^{\alpha}$. All the states in
these subsets are fixed by $a$. Otherwise some state moves to $0$; besides, by the choice of the pair $\{p,q\}$ either $p\dt b^\alpha a=0$ or $q\dt b^\alpha a=0$. In any case we would obtain another pair synchronized by $w$, which is a contradiction with the claim.
Finally, we apply $d-1$ times the word $b^{n-1}a$. Each time the numbers of all states in both subsets decrease by 1. Moreover, through $d-1$ steps no state moves to 0, otherwise we would again get a contradiction with the claim. Thus, $H\dt w\ne S\dt w$.

\textbf{Case 3}: one of the states of the pair $\{p,q\}$ belongs to $I$, and the other to $H\setminus S$. By the claim we have $|S\dt w|=|S|$ and $|H\dt w|=|H|-1$. Let $S\setminus H\ne\varnothing$. The claim implies $0\notin (S\setminus H)\dt w$. By the same argument as in the previous cases we deduce that the sets $(H\setminus S)\dt w$ and $(S\setminus H)\dt w$ do not intersect. So $H\dt w\ne S\dt w$.
If $S\subsetneq H$ and moreover $|S|<|H|-1$, then we get $|S\dt w|=|S|<|H|-1=|H\dt w|$, so $S\dt w\ne H\dt w$.
Finally, it remains to consider the case $H=S\cup\{r\}$. We may assume $r=0$ (otherwise we apply the word $b^{n-r}$ to $S$ and $H$). Let $S=\{q_1,q_2,...,q_\ell\}$, then $H=\{0,q_1,q_2,...,q_\ell\}$. We have $p=0$  and $q=q_i$ for some $1\le i\le \ell$.
%Let $p\in S$, $q=0$.

If $d(0,q)=q$, then we apply the word $u=b^{n-1-q}ab^{q+1}$ to $H$ and $S$. Note that $q\dt u=q+1$, $0\dt u=0$ and for all $1 \leq i \leq \ell$ $q_i \dt u=q_i$ (if $q_i\neq q$). Thus the distance between $0$ and $q$ increases, the distances between $0$ and other states in $S$ are the same, and the distances between $q$ and other states in $S$ decrease.

If $q=n-1$ and $q_1=1$, then we apply the word $u=b^{n-2}ab^{2}$ to $H$ and $S$. Note that $0\dt u=0$, $q_1\dt u=2$, $q\dt u=q$ and for all $2\leq i \leq \ell$ $q_i \dt u=q_i$ (if $q_i\neq q$). Thus the state $q_1=1$ moves to the state $2$, all the other states remain unchanged.

If $d(0,q)=n-q$ and $q_1>1$, then we apply the word $u=b^{n-1}a$ to $H$ and $S$. Note that $0\dt u=0$, for all $ 1\leq i \leq \ell$ $q_i \dt u=q_i-1$. Thus the distance between $0$ and $q$ increases, the distances between $q$ and other states in $S$ remain the same.

Next we construct subsets $H \dt u$ and $S \dt u$. Note that $S\dt u \subseteq H \dt u$ and moreover $H\dt u \setminus S \dt u=\{0\}$. Choose the corresponding pair $\{\overline{p},\overline{q}\}$ for the subsets $S\dt u$ and $H\dt u$. If $\overline{p},\overline{q}\in S\dt u$, then we apply the argument from Case 2 and find a word $\overline{w}$ such that $S\dt uw\ne H\dt uw$ and $|H\dt uw|=|H\dt u|-1= |H|-1$.
Otherwise repeat the algorithm above applied to the subsets $H\dt u$ and $S \dt u$. Through the finite number of steps we will obtain the subsets $H\dt \overline{u}$ and $S\dt \overline{u}$ such that the corresponding pair $\{\overline{p},\overline{q}\}$ is contained in $S \dt \overline{u}$. And this case was studied above.

So we have that for two arbitrary subsets $H$ and $S$ there exists a word $w \in \Sigma^{*}$ such that $H\dt w\neq S \dt w$ and at least one of the equalities $|H \dt w|=|H|-1$ and $|S\dt w|=|S|-1$ holds. Next we consider subsets $H \dt w$, $S \dt w$. If none of them is $\{0\}$, apply described algorithm again. It is clear that through the finite number of steps we will find a word $\overline{w}$ with the property $\overline{w} \in Syn(H) \setminus Syn(S)$ (or $\overline{w} \in Syn(S) \setminus Syn(H)$).
\qed
\end{proof}

Using the same technique as in the previous proposition we prove the same result for the automaton $\mathscr{L}_{n}$ (see Fig.~\ref{fig Ln}). Due to space limits we omit the proof here.
\begin{figure}[ht]
\begin{center}
  \unitlength=2pt
\begin{picture}(60,70)(0,5)
\gasset{Nw=10,Nh=10,Nmr=5}
\thinlines
\node(A1)(0,60){$0$}
\node(A2)(30,75){$1$}
\node[Nframe=n](p0)(60,60){$\ldots$}
\node(n3)(60,25){$n$-$3$}
\node(n2)(30,10){$n$-$2$}
\node(n1)(0,25){$n$-$1$}
\drawedge(A1,A2){$a,b$}
\drawedge(A2,p0){$a,b$}
\drawedge(p0,n3){$a,b$}
\drawedge(n3,n2){$b$}
\drawedge(n2,n1){$b$}
\drawedge(n1,A1){$a,b$}
\drawedge(n3,n1){$a$}
\drawedge(n2,A1){$a$}
\end{picture}
\end{center}
\caption{Automaton $\mathscr{L}_{n}$}
\label{fig Ln}
\end{figure}
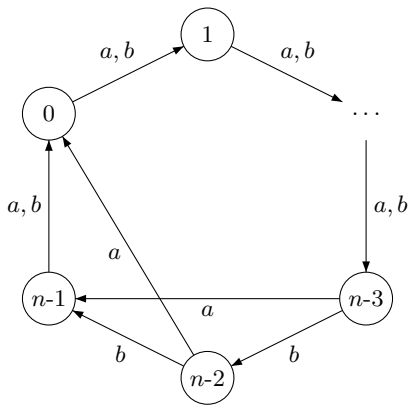

\begin{proposition}
\label{Ln sc}
$sc(Syn(\mathscr{L}_{n}))=2^{n}-n$.
\end{proposition}

Finally, consider the DFA with $n$ states $\mathscr{V}_{n}$ (see Fig.~\ref{fig Vn}).
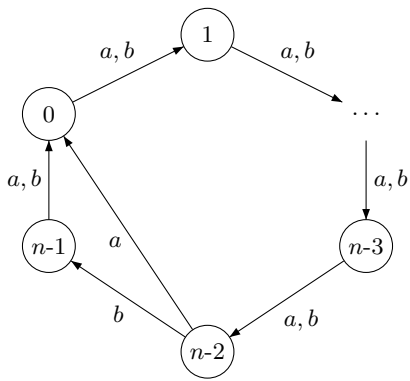
\begin{figure}[ht]
\begin{center}
\unitlength=2pt
\begin{picture}(60,65)
\gasset{Nw=10,Nh=10,Nmr=5}
\thinlines
\node(A1)(0,50){$0$}
\node(A2)(30,65){$1$}
\node[Nframe=n](p0)(60,50){$\ldots$}
\node(n3)(60,25){$n$-$3$}
\node(n2)(30,5){$n$-$2$}
\node(n1)(0,25){$n$-$1$}
\drawedge(A1,A2){$a,b$}
\drawedge(A2,p0){$a,b$}
\drawedge(p0,n3){$a,b$}
\drawedge(n3,n2){$a,b$}
\drawedge(n2,n1){$b$}
\drawedge(n2,A1){$a$}
\drawedge(n1,A1){$a,b$}
\end{picture}
\end{center}
\caption{Automaton $\mathscr{V}_{n}$}
\label{fig Vn}
\end{figure}

\begin{proposition}

$sc(Syn(\mathscr{V}_{n}))=2^{n}-n$.
\end{proposition}
\label{Vn sc}
\begin{proof}
From the DFA $\mathscr{V}_{n}=\langle Q,\Sigma,\delta\rangle$ construct its power automaton ${\cal P}=\langle {\cal Q}, \Sigma,\delta,F\rangle$.
First we check that all nonempty subsets $H\subseteq Q$ are reachable. By induction on $k=|H|$. Case $|H|=k$ is clear: the state set $Q$ of $\mathscr{V}_{n}$ is the initial state of ${\cal P}$. Assume that any subset with cardinality $1<k\leq n$ is reachable. Now we verify that all subsets $H$ with $|H|=k-1$ are reachable. Suppose that $H=\{p_{1},p_{2},...,p_{k-1}\}$ and $p_{i}<p_{i+1}$ for all $1\leq i\leq k-2$.
If $p_{1}\neq 0$, consider the subset $H'=\delta^{-1}(H,b^{p_1})$. Note that $0\in H'$.
If $n-1 \not\in H'$ then $H'$ is reachable from $T=\delta^{-1}(H',a)$, and $|T|=k$. If $n-1 \in H'$ then we can find an integer $\alpha_{1}$ such that $n-1 \not\in \delta^{-1}(H',b^{\alpha_1})$, then find an integer $\alpha_{2}$ such that $0 \in \delta^{-1}(H',a^{\alpha_{2}}b^{\alpha_{1}})$, and $n-1\notin \delta^{-1}(H',a^{\alpha_{2}}b^{\alpha_{1}})$. It is easy to see that $H'$ is reachable from $T=\delta^{-1}(H',aa^{\alpha_{2}}b^{\alpha_{1}})$. Note that also in this case $|T|=k$, so by induction hypothesis $T$ is reachable. Therefore the subset $H'$ is also reachable, thus $H$ is reachable. So the automaton ${\cal P}$ consists of $2^{n}-1$ states. All singletons are equivalent, thus the minimal automaton recognizing $Syn(\mathscr{V}_{n})$ has at most $2^{n}-1-n+1=2^{n}-n$ states.

Next we prove that any two states of ${\cal P}$ which differ from the terminal one are inequivalent. For the proof we use the result from Proposition 2 and the technique from \cite{AnGuVol}. Let $\delta:Q\times \Sigma^{*}\rightarrow Q$ be the transition function of the \v{C}ern\'{y} automaton $\mathscr{C}_{n}$. We transform the \v{C}ern\'{y} automaton by defining a new transition function as follows. Take an arbitrary state $p$ and put: $\delta_{1}(p,b)=\delta(p,b)$, $\delta_{1}(p,c)=\delta(p,ab)$, where $c$ is a new letter. It is not difficult to see that the DFA $\mathscr{C}_{n}$ is transformed to the DFA $\mathscr{V}_{n}'$ (see Fig. ~\ref{fig Vn1}) over the alphabet $\{b,c\}$.
\begin{figure}[ht]
\begin{center}
\unitlength=2pt
\begin{picture}(60,65)
\gasset{Nw=10,Nh=10,Nmr=5}
\thinlines
\node(A1)(0,50){$0$}
\node(A2)(30,65){$1$}
\node[Nframe=n](p0)(60,50){$\ldots$}
\node(n3)(60,25){$n$-$3$}
\node(n2)(30,5){$n$-$2$}
\node(n1)(0,25){$n$-$1$}
\drawedge(A1,A2){$c,b$}
\drawedge(A2,p0){$c,b$}
\drawedge(p0,n3){$c,b$}
\drawedge(n3,n2){$c,b$}
\drawedge(n1,A1){$b$}
\drawedge(n2,n1){$c,b$}
\drawedge(n1,A2){$c$}
\end{picture}
\end{center}
\caption{Automaton $\mathscr{V}_{n}'$}
\label{fig Vn1}
\end{figure}
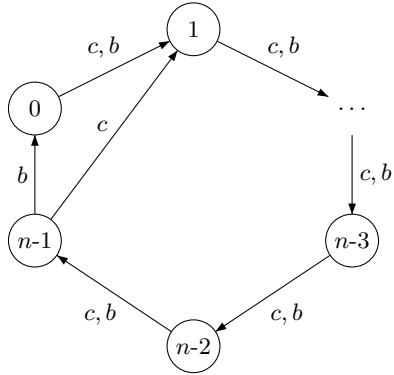
The DFA $\mathscr{V}_{n}'$ is the same as $\mathscr{V}_{n}$ up to renaming letters $c$ and $a$ and the cyclic shift of the state numbers. So we can consider the automaton $\mathscr{V}_{n}'$ instead of $\mathscr{V}_{n}$. Take two arbitrary not equal subsets $H$ and $S$ of the state set of $\mathscr{V}_{n}'$. Since $H$ and $S$ are inequivalent as states in the power automaton ${\cal P}(\mathscr{C}_{n})$, there exists a word $w$ synchronizing only one of them. Since for any subset $T$ in the DFA $\mathscr{C}_{n}$ we have $T\dt aa=T\dt a$, then we can assume that $w$ doesn't contain repeating letters $a$. Thus each occurrence of $a$ in $w$ is either followed by the letter $b$, or is the last letter of $w$. Since $w$ synchronizes only one of the subsets $H$ and $S$, and $b$ is a permutation letter, then $wb$ also synchronizes only one of the subsets $H$ and $S$. So we may assume that all the occurrences of $a$ in $w$ are followed by $b$. We construct the word $\overline{w}$ from $w$ as follows: all inclusions of $ab$ in $w$ replace with $c$. By this construction the word $\overline{w}$ acts on subsets of $\mathscr{V}_{n}'$ in the same way as in $\mathscr{C}_{n}$. So in the automaton $\mathscr{V}_{n}'$ this word synchronizes only one of the subsets $H$ and $S$. Thus the corresponding states of the power automaton of $\mathscr{V}_{n}'$ are not equivalent.
\qed
\end{proof}

\begin{theorem}
$rc(Syn(\mathscr{C}_{n}))=rc(Syn(\mathscr{L}_{n}))=rc(Syn(\mathscr{V}_{n}))=n$.
\end{theorem}
\begin{proof}
Arguing by contradiction suppose that $rc(Syn(\mathscr{C}_{n}))<n$. From the minimal in terms of reset complexity automaton for $Syn(\mathscr{C}_{n})$ construct its power automaton consisting only of reachable subsets. This power automaton have at most inequivalent $2^{n-1}-(n-1)$ states and recognizes $Syn(\mathscr{C}_{n})$. However, by Proposition 2 we have $sc(Syn(\mathscr{C}_{n}))=2^{n}-n$, a contradiction. Thus $rc(Syn(\mathscr{C}_{n}))\geq n$. But $Syn(\mathscr{C}_{n})$ is the language of reset words of the \v{C}ern\'{y} automaton $\mathscr{C}_{n}$ which has exactly $n$ states. Hence, the equality $rc(Syn(\mathscr{C}_{n}))=n$ holds. The other equalities are obtained analogously.
\qed
\end{proof}

Thus, we see that the description of an ideal language $L$ by means of an automaton for which $L$ serves as the language of synchronizing words can be exponentially more succinct than the ``standard'' description via minimal
automaton recognizing $L$.

\section{\normalsize{On uniqueness of the minimal in terms of reset complexity automaton}}

The minimal DFA recognizing a given language is constructed in a unique way up to isomorphism. However, minimal in terms of reset complexity automaton may be constructed in various ways. We give the corresponding example to demonstrate this fact.

Recall that a DFA is called \textit{strongly connected} if for any two states $p,q$ there exists a word $w\in \Sigma^{*}$ such that $\delta (p,w)=q$. The state $s$ of a DFA $\mathscr{A}=\langle Q,\Sigma,\delta\rangle$ is called the \textit{sink} if $s \dt a=s$ for all $a\in\Sigma$.

Here we exhibit a strongly connected 6-state
synchronizing automaton $\mathscr{S}_{6}$ (see Fig. ~\ref{fig S6}) and a 6-state synchronizing automaton
$\mathscr{Z}_{6}$ (see Fig. ~\ref{fig Z6}) having a \emph{sink} state (a state fixed by all letters)
such that $\mathscr{Z}_{6}$ and $\mathscr{S}_{6}$ have the same language of
synchronizing words, namely $L=(a+b)^*(b^{3}ab^{2}a+a^{2}b^{3}a+abab^{3}a+
ab^{2}ab^{3}a)(a+b)^*$. By an exhaustive computer search we have shown that $L$
is not the language of synchronizing words for any synchronizing automaton with
less than 6~states whence both $\mathscr{Z}_{6}$ and $\mathscr{S}_{6}$ are
minimal in terms of reset complexity.

\begin{figure}[ht]
\begin{center}
\begin{picture}(100,30)
   \gasset{Nw=6,Nh=6,Nmr=3}
\thinlines
\node(A)(10,20){$0$}
\node(B)(30,20){$1$}
\node(C)(50,20){$2$}
\node(D)(70,20){$3$}
\node[Nadjust=wh,Nmr=1](E)(70,2){$4$}
\node[Nadjust=wh,Nmr=1](F)(95,20){$5$}
\drawedge(B,C){$b$}
\drawedge(C,D){$b$}
\drawedge(D,E){$b$}
\drawedge[ELside=r](E,F){$b$}
\drawloop[loopangle=120](A){$a,b$}
\drawloop[loopangle=120](B){$a$}
\drawloop[loopangle=0](F){$b$}
\drawedge[ELside=r,curvedepth=-6](C,B){$a$}
\drawedge[curvedepth=7](D,B){$a$}
\drawedge[curvedepth=7](E,A){$a$}
\drawedge[ELside=r,curvedepth=-10](F,C){$a$}
\end{picture}
\end{center}
\caption{Automaton ${\cal Z}_{6}$}
\label{fig Z6}
\end{figure}
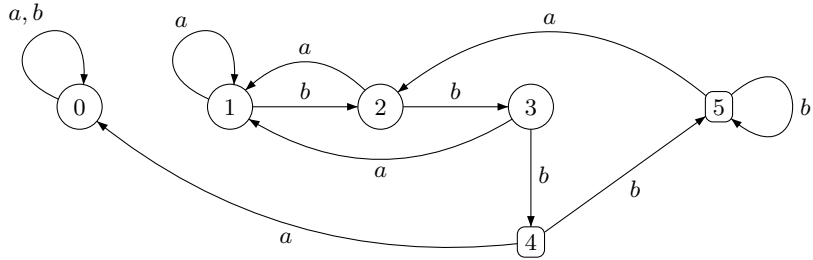

\begin{figure}[ht]
\begin{center}
\begin{picture}(100,30)
\gasset{Nw=6,Nh=6,Nmr=3}
\thinlines
\node(A)(10,20){$0$}
\node(B)(30,20){$1$}
\node(C)(50,20){$2$}
\node(D)(70,20){$3$}
\node[Nadjust=wh,Nmr=1](E)(70,2){$4$}
\node[Nadjust=wh,Nmr=1](F)(95,20){$5$}
\drawedge(A,B){$b$}
\drawedge(B,C){$b$}
\drawedge(C,D){$b$}
\drawedge(D,E){$b$}
\drawedge[ELside=r,curvedepth=-7](F,C){$a$}
\drawedge(E,F){$b$}
\drawloop[loopangle=120](A){$a$}
\drawloop[loopangle=-120](B){$a$}
\drawedge[ELside=r,curvedepth=-7](C,A){$a$}
\drawedge[curvedepth=7](D,B){$a$}
\drawedge[curvedepth=7](E,B){$a$}
\drawedge[curvedepth=7](E,B){$a$}
\drawedge[curvedepth=5](F,E){$b$}
\end{picture}
\end{center}
\caption{Automaton ${\cal S}_{6}$}
\label{fig S6}
\end{figure}
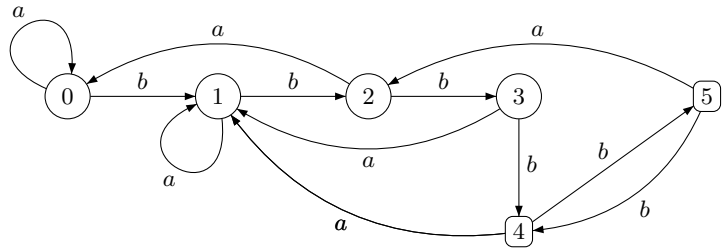

\textbf{Future work}
The question that remains open is whether or not the uniqueness takes places within the class of automata with
sink and within the class of strongly connected automata. Also it would be interesting to design algorithms (and study the computational complexity) for the following problems:

\verb"Find_minimal_reset":\\
\textbf{Input}: A DFA $\mathscr A$ recognizing a language $L$ such that $L=\Sigma^*L\Sigma^*$;\\
\textbf{Output}: A synchronizing DFA $\mathscr{B}$ such that $Syn(\mathscr{B})=L$ and $\mathscr{B}$ has $rc(L)$ states.

\verb"Check_minimal_reset":\\
\textbf{Input}: A synchronizing DFA $\mathscr{A}$;\\
\textbf{Question}: Is $\mathscr{A}$ minimal in terms of reset complexity for the language $Syn(\mathscr{A})$?

The algorithms for these problems might give a hint on how to improve the lower bound on the reset complexity in terms of the shortest word in the language, and in this way approach the \v{C}ern\'{y} conjecture.

\textbf{Acknowledgement} The author thanks lecturer Elena V. Pribavkina  for proposing the problem and for precious suggestions.


\begin{thebibliography}{99}

\bibitem{AnGuVol}
D.S. Ananichev, V.V. Gusev, M.V. Volkov \emph{Slowly Synchronizing Automata
 and Digraphs}, LNCS 6281(010), MFCS 2010, 55-65.
\bibitem{Cerny}
\v{C}ern\'{y} J. \emph{Pozn\'{a}mka k homog\'{e}nnym eksperimentom s kone\v{c}n\'{y}mi automatami},Mat.-Fiz. \v{C}as. Slovensk. Akad. 1964. V.14. P.208-216.[in Slovak]
\bibitem{Pin}
Pin J.-E. \emph{On two combinatorial problems arising from automata theory}, Ann. Discrete Math. 1983. V.17. P. 535-548.
\bibitem{Prib}
Pribavkina E. V., Rodaro E. \emph{Finitely generated synchronizing automata}, In A. H. Dediu, A. M. Ionescu, C. Mart\'{i}n-Vide (eds.) Int. Conf. LATA 2009, Lect. Notes Comp. Sci., Springer-Verlag, Berlin-Heidelberg-New York. 2009. V.5457. P.672-683
\bibitem{Vo_Survey}
M.\,V.~Volkov. \emph{Synchronizing automata and the \v{C}ern\'y
conjecture}. In C.\,Mart\'\i{}n-Vide, F.\,Otto, H.\,Fernau (eds.), Languages
and Automata: Theory and Applications. LATA 2008. Lect.\ Notes Comp.\ Sci.
\textbf {5196} Berlin, Springer, 2008, 11--27.



\end{thebibliography}
\end{document}